\documentclass[a4paper,UKenglish,cleveref, autoref, thm-restate,authorcolumns,final]{lipics-v2019}

\usepackage{comment}
\usepackage{graphicx}
\usepackage{eulervm,amsmath,amssymb}
\usepackage[svgnames]{xcolor}
\usepackage[backgroundcolor=SkyBlue,linecolor=SkyBlue,obeyFinal]{todonotes}

\makeatletter
\providecommand\@dotsep{5}
\def\listtodoname{}
\def\listoftodos{\@starttoc{tdo}\listtodoname}
\makeatother

\newcounter{nmcomment}
\newcommand{\nmtodo}[1]{%
    \refstepcounter{nmcomment}%
    {%
        \todo[color={LightSteelBlue},size=\small,inline]{%
            \textbf{Comment [NM\thenmcomment ]:}~#1}%
}}

\newcounter{hmcomment}

\title{Imbalance Parameterized by Twin Cover Revisited} 

\titlerunning{Imbalance Parameterized by Twin Cover Revisited} 

\author{Neeldhara Misra}{Indian Institute of Technology, Gandhinagar, India}{neeldhara.m@iitgn.ac.in}{}{}

\author{Harshil Mittal}{Indian Institute of Technology, Gandhinagar, India}{mittal\_harshil@iitgn.ac.in}{}{}

\authorrunning{N. Misra and H. Mittal} 

\Copyright{John Q. Public and Joan R. Public} 



\category{} 

\relatedversion{} 

\supplement{}

\nolinenumbers 

\hideLIPIcs  

\EventEditors{John Q. Open and Joan R. Access}
\EventNoEds{2}
\EventLongTitle{42nd Conference on Very Important Topics (CVIT 2016)}
\EventShortTitle{CVIT 2016}
\EventAcronym{CVIT}
\EventYear{2016}
\EventDate{December 24--27, 2016}
\EventLocation{Little Whinging, United Kingdom}
\EventLogo{}
\SeriesVolume{42}
\ArticleNo{23}

\begin{document}

\maketitle

\begin{abstract}
We study the problem of \textsc{Imbalance} parameterized by the twin cover of a graph. We show that \textsc{Imbalance} is XP parameterized by twin cover, and FPT when parameterized by the twin cover and the size of the largest clique outside the twin cover. In contrast, we introduce a notion of succinct representations of graphs in terms of their twin cover and demonstrate that \textsc{Imbalance} is NP-hard in the setting of succinct representations, even for graphs that have a twin cover of size one.  
\end{abstract}

\section{Introduction}
\label{sec:intro}
Graph layout problems are combinatorial optimization problems where the objective is to find a permutation of the vertex set that optimizes some function of interest. In this paper we focus on the problem of determining the \emph{imbalance} of a graph $G$. Given a permutation $\pi$ of $V$, we define the left and right neighborhood of a vertex, $N_L(v)$ and $N_R(v)$, to be the number of vertices in the neighborhood of $v$ that were placed before and after $v$ in $\pi$. The imbalance of $v$ is defined as the absolute difference between these quantities, that is, $\big||N_L(v)|-|N_R(v)|\big|$ and the imbalance of the graph $G$ with respect to $\pi$ is simply the sum of the imbalances of the individual vertices. The imbalance of $G$ is the minimum imbalance of $G$ over all orderings of $V$, and an ordering yielding this imbalance is called an optimal ordering\footnote{We refer the reader to the preliminaries for formal definitions of the terminology that we use in this section.} The problem was introduced by~\cite{BiedlCGHW05}, and finds several applications, especially in graph drawing~\cite{Kant96,KantH97,PapakostasT98,Wood03,Wood04}.

\textsc{Imbalance} is known to be NP-complete for several special classes of graphs, including bipartite graphs of maximum degree six~\cite{BiedlCGHW05} and for graphs of degree at most four~\cite{KaraKW07}. Further, the problem is known to be FPT when parameterized by imbalance~\cite{LokshtanovMS13}, vertex cover~\cite{FellowsLMRS08}, neighborhood diversity~\cite{Bakken18}, and the combined parameter treewidth and maximum degree~\cite{LokshtanovMS13}. Recently, it was claimed that imbalance is also FPT when parameterized by twin cover~\cite{GorznyB19}, which is a substantial improvement over the vertex cover parameter. We mention briefly here that a vertex cover of a graph $G$ is a subset $S \subseteq V(G)$ such that $G \setminus S$ is an independent set, and a twin cover of a graph is a subset $T \subseteq V(G)$ such that the connected components of $G \setminus S$ consist of vertices which are true twins --- in particular, note that each connected component induces a clique, and further, all vertices have the same neighborhood in the cover $T$. The method employed in~\cite{GorznyB19} to obtain a FPT algorithm for \textsc{Imbalance} parameterized by twin cover relies on a structural lemma which, roughly speaking, states that there exist optimal orderings where any maximal collection of true twins appear together. Based on this, it is claimed that the cliques of $G \setminus S$ can be contracted into singleton vertices to obtain an equivalent instance $H$. By observing that the twin cover of $G$ is a vertex cover of $H$, we may now use the FPT algorithm for \textsc{Imbalance} parameterized by vertex cover to obtain an imbalance-optimal ordering for $H$, and the contracted vertices can be ``expanded back in place'' to recover an optimal ordering for $G$. 

Although the structural lemma is powerful, unfortunately, we are unable to verify the safety of the contraction step based on it. We note that the graph $H$ has lost information about the sizes of the individual cliques from $G$, and the ILP formulation is blind to the distinctions between vertices that correspond to cliques of different sizes. Consider, for example, an instance with a twin cover of size one --- the reduced instance is a star and any layout of $H$ that distributes the leaves of the star almost equally around the center would have the same imbalance, and these are indeed all the optimal layouts of $H$. On the other hand, several of these layouts could have different imbalances when considered in the context of $G$. A natural fix to this issue is to mimic the ILP formulation directly for the graph $H$, and taking advantage of the structural lemma to come up with an appropriate set of variables that correspond to cliques of a fixed size. Unfortunately, this leads us to a situation where the number of variables is a function of the sizes of the cliques, and only yields an algorithm that is FPT in the size of the twin cover and the size of the largest clique outside the twin cover. 
\vspace{-10pt}
\paragraph*{Our Contributions.} With the premise that there is more to the twin cover parameterization for \textsc{Imbalance}, the focus of this paper is on the complexity of \textsc{Imbalance} parameterized by twin cover. We demonstrate that the problem is in XP when parameterized by twin cover (Theorem~\ref{thm:xp}), and FPT when parameterized by the twin cover and the size of the largest clique (Theorem~\ref{thm:fpt}). The first result is based on a slightly non-trivial dynamic programming algorithm, which can be thought of as a generalization of the classic dynamic programming routine for the Partition problem, in which we are given $n$ numbers and the question is if they can be partitioned into two groups of equal sum. Indeed, the approach is inspired by the fact that the problem of finding the optimal layout for graphs that have a twin cover of size one is essentially equivalent to the Partition problem. However, generalizing to larger-sized twin covers involves accounting for several details and we also discuss why the natural brute-force approaches to an XP algorithm end up failing. The second result is based on the ILP formulation that we alluded to earlier. We mention here that we rely crucially on the structural result of~\cite{GorznyB19} for arguing the correctness of our algorithmic approaches.

We also propose a natural notion of a succinct representation for graphs with bounded twin covers. Note that such graphs can be specified completely by the adjacency matrix of the twin cover and for each clique outside the twin cover, its size and its neighborhood in the twin cover. In contrast to the algorithmic results above, the resemblance to Partition leads us to the curious observation that the problem of \textsc{Imbalance} is NP-hard even for graphs that have a twin cover of size one in the succinct setting. This is formalized in Theorem~\ref{thm:paranphard}. We find it particularly interesting that \textsc{Imbalance} is a problem for which the choice of representation leads to a stark difference in the complexity of the problem. 

We note here that several FPT algorithms for other problems parameterized by twin cover remain FPT in a ``strongly polynomial'' sense and can be easily adapted to the setting of succinct input. For example, the algorithm for \textsc{Equitable Coloring}~\cite{Ganian11} relies on reducing the problem to a maximum flow formulation where the sizes of the cliques outside the twin cover feature as capacities in the flow network, and this approach remains efficient for succinct input since the maximum flow can be found in strongly polynomial time. Similarly, the problems of \textsc{PreColoring Extension}, \textsc{Chromatic Number}, \textsc{MaxCut}, and \textsc{Feedback Vertex Set} as proposed by Ganian~\cite{Ganian11} can be easily adapted to being strongly polynomial in our succinct representation.

This paper is organized as follows. We begin by describing the notation and terminology that is the most relevant to our discussions in the next section, and refer the reader to~\cite{CyganFKLMPPS15} for background on the parameterized complexity framework and to~\cite{DiazPS02} for a survey of graph layout problems. We then establish the para-NP-hardness of \textsc{Imbalance} parameterized by twin cover in the setting of succinct representations in Section~\ref{sec:partition}. The XP algorithm when parameterized by the twin cover, and the FPT algorithm in the combined parameter of twin cover and largest clique size, are described in Sections~\ref{sec:xp} and~\ref{sec:ilp}, respectively.  

\section{Preliminaries}
\label{sec:prelims}
We use $G = (V,E)$ to denote an undirected, simple graph unless mentioned otherwise, and we will typically use $n$ and $m$ to denote $|V|$ and $|E|$, respectively. The \emph{neighborhood} of a vertex $v \in V$ is given by $N(v) := \{u ~|~ (u,v) \in E\}$. The \emph{closed neighborhood} of a vertex $v$ is given by $N[v] := N(v) \cup \{v\}$. Likewise, the open and closed neighborhoods of a set $S \subseteq V$ are defined as: $N(S) := \{v ~|~ v \notin S \mbox{ and } \exists u \in S \mbox{ such that } (u,v) \in E\} \mbox{ and } N[S] := N(S) \cup S,$ respectively. A subset $Y \subseteq V$ is said to be a set of true twins in $G$ if for every pair of vertices $u,v$ in $Y$, $N[u]=N[v]$.

Let $\mathcal{S}(V)$ denote the set of all orderings of $V(G)$, and let $\sigma$ be an arbitrary but fixed ordering of $V(G)$. For $1 \leq i \leq n$, $i^{th}$ vertex in the ordering is denoted by $\sigma(i)$. The relation $<_\sigma$ is defined as $u <_\sigma v$ if and only if $u$ precedes $v$ in the ordering $\sigma$. We also define the \emph{left neighborhood} and \emph{right neighborhood} of a vertex $v$ in the natural way:

$$ N_L(v,\sigma) = \{ u ~|~ u \in N(v) \mbox{ and } u <_\sigma v\} \mbox{ and } N_R(v,\sigma) = \{ u ~|~ u \in N(v) \mbox{ and } v <_\sigma u\} .$$

We also use $p(v,\sigma)$ and $q(v,\sigma)$ to denote the sizes of $N_L(v,\sigma)$ and $N_R(v,\sigma)$, respectively, and we refer to these numbers as the \emph{predecessors} and the \emph{successors} of $v$ with respect to $\sigma$. If the permutation $\sigma$ is clear form the context, we use the terms predecessors and successors without qualifying for $\sigma$. 

An ordering $\sigma$ of $V$ is said to be a \emph{clean ordering} if for every inclusion-wise maximal subset $Y \subseteq V$ that forms a set of true twins in $G$, the vertices of $Y$ appear consecutively in $\sigma$, i.e., $Y$ contains all vertices in $V$ that lie between the smallest and largest elements of $Y$ (with respect to $<_\sigma$).

\textbf{Imbalance.} The \emph{imbalance} of a vertex $v$ with respect to an ordering $\sigma$ of $V$ is denoted $\mathcal{I}(v,\sigma)$, and is defined as the absolute difference between the predecessors and the successors of $v$, that is, $\mathcal{I}(v,\sigma) = |p(v,\sigma) - q(v,\sigma)|$. The imbalance of an ordering $\sigma$, denoted $\mathcal{I}(\sigma)$, is the total imbalance of all the vertices with respect to $\sigma$, and the imbalance of the graph $G$ is minimum imbalance over all permutations of $V$:

$$ \mathcal{I}(G) = \min_{\sigma \in \mathcal{S}(V)} \mathcal{I}(\sigma), \mbox{ where } \mathcal{I}(\sigma) = \sum_{v \in V} \mathcal{I}(v,\sigma).$$

An ordering $\sigma$ of $V$ is said to be an \emph{imbalance optimal ordering} if $\mathcal{I}(\sigma)=\mathcal{I}(G)$. We recall the following fact from~\cite{GorznyB19}.

\begin{lemma}[\cite{GorznyB19}]\label{lem:truetwinsstaytogether} There exists a clean imbalance optimal ordering. 
\end{lemma}

\textbf{Twin Cover.} A subset $S \subseteq V$ is called a twin cover if for every component $X$ of $G \setminus S$, vertices of $X$ form a set of true twins in $G$. In other words, for every component $X$ of $G\setminus S$,vertices of $X$ form a clique such that for every pair of vertices $u,v$ in $V(X)$, $N(u) \cap S = N(v) \cap S$. Henceforth, we will refer to the maximal cliques,or equivalently the components, of $G\setminus S$ as simply the 'cliques' of $G\setminus S$ for the sake of simplicity. We also say that $S \subseteq V$ is a \emph{$\ell$-bounded twin cover} if it is a twin cover such that every clique in $G \setminus S$ has at most $\ell$ vertices. 


Note that the imbalance of a layout does not change if the positions of any pair of true twins are exchanged. Therefore, the following is an immediate consequence of Lemma~\ref{lem:truetwinsstaytogether}.

\begin{corollary}
Let $G$ be a graph and let $S \subseteq V(G)$ be a twin cover of $G$. Then, there exists an imbalance optimal ordering of $G$ where the vertices of every clique in $G \setminus S$ appear consecutively. 
\end{corollary}

For further discussion, a clean ordering in the context of a graph $G$ given with a twin cover $S$ is understood to be an ordering in which the vertices of every clique of $G \setminus S$ appear consecutively. Further, we also abuse language slightly and use the term ``cliques'' to always refer to the maximal cliques of $G \setminus S$, unless mentioned otherwise. 

For a graph $G$ with twin cover $S$, we define the \emph{type} of a vertex $v$ in $G \setminus S$ as $N(v) \cap S$, and the type of a clique $C$ in $G \setminus S$ as the type of any arbitrarily chosen vertex in $C$. Observe that all vertices of any clique $C \in G \setminus S$ have the same type, and therefore the notion of the type of a clique is well-defined. Note that $G$ is completely specified once the structure of a twin cover $S$ and the sizes and types of all the cliques in $G \setminus S$ is given. 

In particular, given $\mathcal{G} := (H,\{(\ell_i,S_i)~|~ 1 \leq i \leq r\})$, where each $S_i$ is a subset of $V(H)$, the graph $G$ associated with $\mathcal{G}$ is defined in the following natural way. The vertex set of $G$ is given by $V(G) := S \cup C_1 \cup \cdots C_i \cup \cdots \cup C_r,$ where $|C_i| = \ell_i$ for all $i \in [r]$ and $|S| = |V(H)|$. Now, identify the vertices of $S$ with $V(H)$ via an arbitrary but fixed mapping $f$, and define the set of edges as follows. For any pair of vertices $u,v \in S$, we have $(u,v) \in E(G)$ if and only if $(f(u),f(v)) \in E(H)$. Further, we induce a clique on each $C_i$, and finally, for any clique $C_i$ and a vertex $v \in S$ we add edges between $v$ and every vertex of $C_i$ if and only if $f(v) \in S_i$. We say that $\mathcal{G}$ provides a  \emph{succinct representation based on a twin cover}. For brevity, we will usually refer to $\mathcal{G}$ as a succinct representation of $G$. For further discussion,we use the same notation to refer to both a vertex in $V(H)$ and its preimage (under the function $f$) in $S$, i.e., for any $w$ in $V(H)$, we refer to $f^{-1}(w)$ in $S$ as $w$ for the sake of simplicity.

We now introduce the natural computational question associated with \textsc{Imbalance}. Given a graph $G = (V,E)$, a twin cover $S \subseteq V$ of size at most $k$, and a target $t$, determine if $\mathcal{I}(G) \leq t$. Unless mentioned otherwise\footnote{We also consider the parameter $(k+\ell)$ when we are given a $\ell$-bounded twin cover.}, our focus will be on \textsc{Imbalance} parameterized by $k$, the size of the twin cover. For the most part, we assume that the input graph $G$ is specified in the standard way, i.e, by its adjacency matrix or adjacency list. If, on the other hand, $G$ is specified by a succinct representation in terms of its twin cover, then we say that the input is succinct, and if this is the scenario we are in, we state it explicitly. 



\section{Weak Para-NP-Hardness}
\label{sec:partition}

In this section, we establish the NP-hardness of \textsc{Imbalance} when the input is succinct, even for graphs that have a twin cover of size one. This can be interpreted as a ``weak'' para-NP-hardness result for \textsc{Imbalance} when parameterized by twin cover.  

\begin{theorem}\label{thm:paranphard}
For succinct input, \textsc{Imbalance} is NP-hard even for graphs that have a twin cover of size one.
\end{theorem}

We establish this result by a reduction from \textsc{Partition} problem, which is well-known to be weakly NP-hard. Recall that the input to \textsc{Partition} is a set of positive integers $\{a_1, \ldots, a_r\}$, and the question is if there exists a subset $S \subset [r]$ such that $\sum_{i \in S}a_i = \sum_{i \notin S}a_i$. An intuitive visual for graphs that have a twin cover of size one is to imagine that we have balls suspended from a single point of varying weights, proportional to the sizes of the cliques, and a layout that optimizes the imbalance is faced with the task of distributing these balls on either side of the suspension point so that the total weight on either side is equally distributed. To formalize this idea, we first argue a lower bound for the imbalance of any graph that has a twin cover of size one. To begin with, consider the following function:

 \begin{equation*}
    \gamma(\ell) :=
    \begin{cases}
      \ell^2/2 & \text{if } \ell \text{ is even},\\
      (\ell^2-1)/2 & \text{if } \ell \text{ is odd}.
    \end{cases}
\end{equation*}

We define the \emph{intrinsic imbalance} of a clique $C$ on $\ell$ vertices as $\gamma(\ell)$. Our first claim is the following. 

\begin{proposition}
Let $G$ be given by $\mathcal{G} = (H,\{(\ell_i,S_i)~|~ 1 \leq i \leq r\})$. Then:
$$\mathcal{I}(G) \geq \left(\sum_{i=1}^r \gamma(\ell_i)\right).$$ 
\end{proposition}

\begin{proof}
Let $\sigma$ be an arbitrary clean imbalance optimal ordering of $V$. For every $i$ in $[r]$, let $Y_i$ be the inclusion-wise maximal set of true twins containing the vertices of $C_i$. The vertices of $Y_i$ appear consecutively in $\sigma$.Changing the inner order of true twins in a layout has no effect on imbalance. So without loss of generality, we can assume  that the vertices of $C_i$ appear consecutively in $\sigma$. Label them in the order of their appearance as $v^1_i <_\sigma .....<_\sigma v^{\ell_i}_i$. For every $j$ in $[\ell_i]$, $q(v^j_i, \sigma) = (\ell_i-j) + q(v^{\ell_i}_i, \sigma)$ and $p(v^j_i, \sigma)= (j-1) + p(v^1_i,\sigma)$. Now,
$$\mathcal{I}(G)= \mathcal{I}(\sigma)\geq \sum_{i=1}^{r} \sum_{j=1}^{\ell_i} \mathcal{I}(v^j_i,\sigma)= \sum_{i=1}^{r} \sum_{j=1}^{\ell_i} \left|q(v^{\ell_i}_i,\sigma)- p(v^1_i,\sigma) +(\ell_i +1)-2j \right|$$
For every $i$ in $[r]$,
$$\sum_{j=1}^{\ell_i} |q(v^{\ell_i}_i,\sigma)- p(v^1_i,\sigma) +(\ell_i +1)-2j|\geq \sum_{j=1}^{\ell_i} |(\ell_i+1) -2j|$$
If $\ell_i$ is even, $$\sum_{j=1}^{\ell_i} \left|(\ell_i+1)-2j \right| = 2\sum_{j=1}^{\frac{\ell_i}{2}} (2j-1) = \frac{\ell_i^2}{2}$$
If $\ell_i$ is odd, $$\sum_{j=1}^{\ell_i} \left|(\ell_i+1)-2j \right| = 2\sum_{j=1}^{\frac{\ell_i-1}{2}} 2j = \frac{\ell_i^2 - 1}{2}$$
Hence, $\mathcal{I}(G) \geq \sum_{i=1}^r \gamma(\ell_i)$.

\end{proof}

For a graph $G$ given by $\mathcal{G} := (H,\{(\ell_i,S_i)~|~ 1 \leq i \leq r\})$, define:

$$\iota(G) := \left(\sum_{i=1}^r \gamma(\ell_i)\right).$$ 

The following observation is based on the fact that if a graph has a twin cover of size one, then final imbalance of odd-sized cliques is one more than their intrinsic imbalance, and the final imbalance of even-sized cliques is equal to their intrinsic imbalance, and this is true for any clean ordering, irrespective of where the cliques are placed in the layout relative to the twin cover vertex. 

\begin{proposition}
\label{prop:imbtc1}
Let $G$ be a connected graph given by $\mathcal{G} = (H,\{(\ell_i,S_i)~|~ 1 \leq i \leq r\})$ where $H = \{v\}$. Then $\mathcal{I}(G) \geq \iota(G) + \sum_{i=1}^{r} (\ell_i ~mod~ 2)$.
\end{proposition}

\begin{proof}
We proceed in exactly the same way as in proof of Proposition 1 to obtain $\mathcal{I}(G)\geq \sum_{i=1}^{r} \sum_{j=1}^{\ell_i} \left|q(v^{\ell_i}_i,\sigma)- p(v^1_i,\sigma) +(\ell_i +1)-2j\right|$. Since $|H|=1$ and $G$ is a connected graph,for every $i$ in $[r]$, either $p(v^1_i,\sigma)=1,q(v^{\ell_i}_i,\sigma)=0$ or $p(v^1_i,\sigma)=0,q(v^{\ell_i}_i,\sigma)=1$. In both cases, we have:
$$\sum_{j=1}^{\ell_i} \left|q(v^{\ell_i}_i,\sigma)- p(v^1_i,\sigma) +(\ell_i +1)-2j\right| = \left(\sum_{j=1}^{\ell_i} \left|(\ell_i +1)-2j\right|\right) + (\ell_i ~mod~ 2)$$
Hence, $\mathcal{I}(G)\geq \sum_{i=1}^{r} (\gamma(\ell_i) + (\ell_i ~mod~ 2)) = \iota(G) + \sum_{i=1}^{r} (\ell_i ~mod~ 2)$. 

\end{proof}

We have the following straightforward consequence of Proposition~\ref{prop:imbtc1}.

\begin{corollary}
\label{cor:imbtc1}
Let $G$ be a connected graph given by $\mathcal{G} = (H,\{(\ell_i,S_i)~|~ 1 \leq i \leq r\})$ where $H = \{v\}$. Then $\mathcal{I}(G) = \iota(G) + \sum_{i=1}^{r} (\ell_i ~mod~ 2) + \mathcal{I}(v)$.
\end{corollary}

We are now ready to describe the reduction from \textsc{Partition}.

\nmtodo{May want to simplify the use of $f^{-1}$ by some clarification or abuse of notation declared upfront.}

\begin{proof} \emph{(Proof of Theorem~\ref{thm:paranphard}.)} Given an instance $\mathcal{P} := \{a_1, \ldots, a_r\}$ of \textsc{Partition}, let $G_\mathcal{P}$ be given by $\mathcal{G}(\mathcal{P}) = (H,\{(\ell_i,S_i)~|~ 1 \leq i \leq r\})$, which in turn is defined as follows: $H = \{v\}, \ell_i = a_i \mbox{ for all } i \in [r],\mbox{ and }S_i = \{v\}\mbox{ for all }i \in [r].$

The instance of \textsc{Imbalance} is now given by $(G_\mathcal{P}, S = \{v\}, t = \iota(G_\mathcal{P})+\sum_{i=1}^{r} (a_i ~mod~ 2))$. This completes the construction, and we now turn to a proof of equivalence.

\emph{Forward Direction.}
Suppose there exists a subset $A\subset[r]$ such that $\sum_{i\in A} a_i = \sum_{i\in [r]\setminus A} a_i$. Let $\mathcal{L} := \{C_i ~|~ i \in A\}$ and $\mathcal{R} := \{C_i ~|~ i \in [r] \setminus A\}$.
Consider an arbitrary clean ordering $\sigma$ of $V(G_\mathcal{P})$ that places the vertices of cliques in $\mathcal{L}$ to the left of $v$ and the vertices of cliques in $\mathcal{R}$ to the right of $v$,i.e.,$\sigma$ is an arbitrary clean ordering such that
\begin{itemize}
    \item for every clique $X$ in $\mathcal{L}$, for every vertex $u$ in $X$, $u<_{\sigma} v$
    \item for every clique $Y$ in $\mathcal{R}$, for every vertex $w$ in $Y$, $w>_{\sigma} v$
\end{itemize}
Then, observe that:

$$\mathcal{I}(v) = \left|\sum_{i\in A} a_i - \sum_{i\in [r]\setminus A} a_i \right| = 0,$$

and by Corollary~\ref{cor:imbtc1}, we have that $\mathcal{I}(G_\mathcal{P}) = \iota(G_\mathcal{P}) + \sum_{i=1}^{r} (a_i ~mod~ 2)$, as desired. 


\emph{Reverse Direction.} 
Suppose $\mathcal{I}(G_\mathcal{P})\leq \iota(G_\mathcal{P}) + \sum_{i=1}^{r} (a_i ~mod~ 2)$. Using Proposition $2$, we also know that $\mathcal{I}(G_\mathcal{P})\geq \iota(G_\mathcal{P}) + \sum_{i=1}^{r} (a_i ~mod~ 2)$. Thus, $\mathcal{I}(G_\mathcal{P})= \iota(G_\mathcal{P}) + \sum_{i=1}^{r} (a_i ~mod~ 2)$. Let $\sigma$ be an arbitrary clean imbalance optimal ordering of $V(G_\mathcal{P})$. 

Without loss of generality, vertices of $C_i$ appear consecutively in $\sigma$, for every $i$ in $[r]$. Let:

\begin{itemize}
    \item $\mathcal{L} := \{ i\in[r] ~|~$ for every vertex $u$ in $C_i, u<_\sigma v \}$
    \item $\mathcal{R} := [r]\setminus \mathcal{L} = \{ i\in[r] ~|~$ for every vertex $w$ in $C_i, w >_\sigma v \}$.
\end{itemize} 

Recall that by Corollary~\ref{cor:imbtc1}, we have that $\mathcal{I}(v,\sigma) = 0$.  Thus, we have:
$$\mathcal{I}(v,\sigma) = \left|\sum_{i\in \mathcal{L}} a_i - \sum_{i\in \mathcal{R}} a_i \right| = 0.$$
Hence, $\sum_{i\in \mathcal{L}} a_i = \sum_{i\in \mathcal{R}} a_i$, as desired.
\end{proof}

The argument above establishes the NP-hardness of \textsc{Imbalance} for succinct input, even on graphs that have a twin cover of size one. We can also establish membership in NP as follows. Observe that the imbalance of a layout is completely determined by the imbalance of the vertex $v$, and therefore it suffices to provide a partition of the cliques as a certificate. The layout corresponding to a partition of the numbers $\{\ell_i~|~ 1 \leq i \leq r\})$ into $(L,R)$ would place all cliques whose sizes correspond to numbers in $L$ to the left of the twin cover vertex, and the remaining to the right. Note that it is straightforward to compute the imbalance of the twin cover vertex in the layout associated with the given partition. Also, since the quantity $\iota(G) + \sum_{i=1}^{r} (\ell_i ~mod~ 2)$ can be computed efficiently, we can check if this layout has the desired imbalance in strongly polynomial time. Finally, note that any \textsc{Yes}-instance of the problem admits a valid certificate since there exists an optimal layout that is also clean. We generalize this argument for twin covers of arbitrary size in Section~\ref{sec:xp}.



\section{An XP Algorithm}
\label{sec:xp}
In this section, our goal is to demonstrate an XP algorithm for \textsc{Imbalance} parameterized by twin cover when the entire input is given explicitly in the standard form. Throughout, we use $G$ to denote the graph given as input, $S \subseteq V(G)$ denotes a twin cover of size $k$, and the question is if $G$ admits a layout whose imbalance is at most $t$. We denote the cliques of $G \setminus S$ by $C_1, \ldots, C_r$, and we let $\ell_i$ denote the number of vertices in $C_i$. We also use $\ell$ to denote $\max\{\ell_1, \ldots, \ell_r\}$. 

To begin with, let us consider some natural brute-force approaches, described informally, that will eventually motivate the definitions that we will encounter later.  Assuming we are dealing with a \textsc{Yes}-instance, let $\sigma^\star$ denote a clean ordering of $V(G)$ whose imbalance is at most $t$. 

\paragraph*{A First Approach.} First, we may guess the relative order of the vertices of $S$ in $\sigma^\star$. Further, since $\sigma^\star$ is a clean ordering, we know that for any $C_i \in G \setminus S$ every vertex of $C_i$ lies between some consecutive pair of vertices from $S$ in $\sigma^\star$. Thus, for every clique $C_i \in G \setminus S$, we may guess it's ``final location'', i.e, the choice of the consecutive pair of twin cover vertices that the clique is sandwiched between in $\sigma^\star$. Since there are $k+1$ possible addresses for this final location, this requires examining $O((k+1)^r)$ possibilities, which is too expensive. 

\paragraph*{Using Types Based on Neighbhorhoods.} Recall that the type of a clique is defined as its neighborhood in the twin cover $S$. While there are an unbounded number of cliques in $G \setminus S$, there are at most $2^k$ types of cliques, and we might hope that cliques of the same type may be treated similarly. In particular, consider that for each location between consecutive twin cover vertices in $\sigma^\star$, we guess the number of cliques of each type that appear in that location. For a type $\tau$, let $r(\tau)$ denote the number of cliques of type $\tau$. In contrast to the brute force approach above, this would only require us to examine 

$$\prod_{\tau \in 2^S}{k+r(\tau) \choose k} = O((n^k)^{2^k})$$ 

possibilities. Unfortunately, this guess is too ``coarse'', and does not give us enough information to determine the imbalance of the layout. For example, consider a connected graph with a twin cover of size one. Here every clique has the same type (say $\tau$) and there are two locations to consider in $\sigma^\star$. Suppose we guess that we have $\frac{r}{2}$ cliques of type $\tau$ in the first location and $\frac{r}{2}$ cliques of type $\tau$ in the second location. It is easy to devise examples for which there would be multiple layouts that are consistent with this guess, each with a different imbalance. The main reason for this is that our notion of types here ignores the sizes of the cliques, which is an important ingredient in determining the exact imbalance of the layout in question. This is exemplified with a concrete example in Figure~1. 


\begin{figure}[htp]
\label{fig:eg}
    \centering
    \includegraphics[width=6cm]{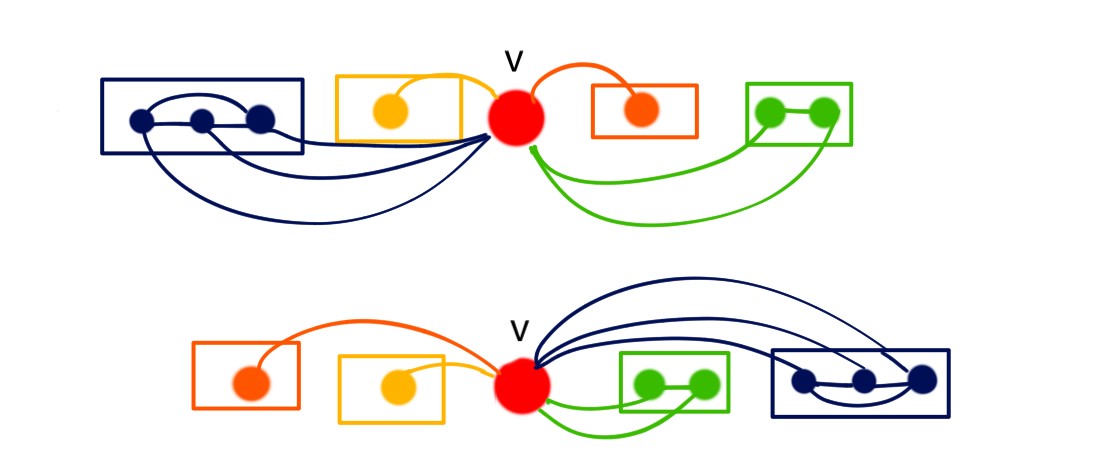}
    \caption{In first layout, $C_1$ and $C_2$ are placed to the left of $v$ and $C_3$ and $C_4$ are placed to the right of $v$ and $\mathcal{I}(v)=1$. In the second layout, $C_3$ and $C_2$ are placed to the left of $v$ and $C_4$ and $C_1$ are placed to the right of $v$ and $\mathcal{I}(v)=3$.}
\end{figure}


\paragraph*{Refining Types.} A natural attempt to fix the previous approach is to refine the notion of types so that they incorporate some information about the sizes of the cliques. To this end, for each $j \in [\ell]$ and each subset $T \subseteq S$ we let $\lambda(T,j)$ denote the set of cliques of type $T$ that have size $j$. Further, we say that any clique in $\lambda(T,j)$ has a \emph{supertype} $\tau(T,j)$. This refinement, used in the framework above, does give us all the information that we need to compute the imbalance of any layout that is consistent with our guesses. On the other hand, our guesses are once again too expensive, since the size of the largest clique now features in the exponent, and is unbounded in the parameter. In particular, note that the number of possibilities that we have to examine has now increased to:

$$\prod_{(T,j) \in 2^S \times [\ell]}{k+|\lambda(T,j)| \choose k} = O((n^k)^{\ell 2^k}).$$ 

Note that although our upper bound is not tightly argued, there are graphs\footnote{Consider, for instance, the running time that is obtained when the input graph has one clique each of size one, two, three, and so on up to $O(\sqrt{n})$.} for which the running time of this algorithm would run into $\Omega(k^{\sqrt{n}})$, and therefore, we may rule out the hope for a XP bound based on an improved analysis of this algorithm.


\paragraph*{Types with Thresholds: The Best of Both Worlds?} To overcome the challenges above, we introduce a type concept that distinguishes between ``small'' and ``large'' cliques. The key observation that leads us to determine a threshold is formalized in Corollary~\ref{cor:largecliques}, but informally, it is the following idea. If two cliques $C_i$ and $C_j$ are sufficiently large and happen to have the same parity \emph{and} the same type $\tau$, but possibly different sizes (that is, $\ell_i \neq \ell_j$), then the total imbalance of their vertices in a layout $\sigma$ can be thought of as $\gamma(\ell_i) + c_\tau(x)$ and $\gamma(\ell_j) + c_\tau(y)$, where $c_\tau(\cdot)$ is a function that depends only on the location of the cliques. In particular, this means that large cliques that have the same type and parity, and which end up in the same location of a layout will have the same ``excess imbalance'' over and above their ``intrinsic imbalance'' (which is given by $\gamma$). This eventually allows us to estimate the \emph{total} imbalance of large cliques by focusing only on the excess to begin with, and finally adding the intrinsic imbalances of all large cliques across the board. 


With this observation at hand, it is tempting to repeat the brute force approach from above with this more nuanced notion of types. However, some reflection reveals that while the approach would yield a XP running time, it is again somewhat information-starved: while Lemma~\ref{lem:largecliques} will allow us to estimate the imbalance of all the cliques in any layout that respects a fixed guess, we would still face ambiguity in determining the imbalances of the twin cover vertices. Therefore, instead of employing the brute-force approach, we turn to a dynamic programming routine, which is inspired by the well-known pseudo polynomial algorithm for the \textsc{Subset Sum} problem. Indeed, for the special case when the input is a connected graph with a twin cover of size one, our algorithm coincides with the \textsc{Subset Sum} DP. 

Before describing our DP table and the associated recurrence, we first introduce some definitions. A clique $C$ in $G \setminus S$ is said to be a \emph{large} clique if $|C| > k$, and is said to be a \emph{small} clique otherwise. Further, a clique $C \in G \setminus S$ is \emph{even} (respectively, \emph{odd}) if it has an even (respectively, odd) number of vertices. Recall that the intrinsic imbalance of a clique on $\ell$ vertices is $\gamma(\ell)$. We now introduce two related notions. The \emph{total imbalance} and \emph{excess imbalance} of a clique $C$ on $\ell$ vertices with respect to a layout $\sigma$ is given by, respectively: 

$$ \gamma^\star(C,\sigma) = \sum_{v \in C} \mathcal{I}(v,\sigma) \text{ and } \gamma^+(C,\sigma) = \gamma^\star(C,\sigma) - \gamma(\ell). $$


We now claim that the excess imbalance of any large clique is a function of its parity, type, and location in the layout $\sigma$, and in particular, is independent of its size.

\begin{lemma}\label{lem:largecliques} 
Let $G$ be a graph and let $S \subseteq V(G)$ be a twin cover of $G$ of size $k$. Further, let $C$ be a large clique in $G \setminus S$ of type $T \subseteq S$. For any layout $\sigma$ we have:

 \begin{equation*}
    \gamma^+(C,\sigma) =
    \begin{cases}
      \lfloor \frac{\delta(C,\sigma)^2}{2} \rfloor & \text{if } C \text{ is an even clique},\\
      \lceil  \frac{\delta(C,\sigma)^2}{2} \rceil & \text{if } C \text{ is an odd clique},\\
    \end{cases}
\end{equation*}

where $\delta(C,\sigma)$ denotes the difference $|N_L(v,\sigma) \cap S| - |N_R(v,\sigma) \cap S|$, for any $v \in C$.
\end{lemma}

\begin{proof} To begin with, note that the total imbalance of $C$ is given by:

$$\gamma^{\star}(C,\sigma)=\sum_{j=1}^{|C|}\Big| \delta(C,\sigma) + 2j - |C| - 1 \Big| = \sum_{j=1}^{|C|}\Big| \delta(C,\sigma) - 2j + |C| + 1 \Big|,$$

where the second expression is obtained by adding the imbalances of the vertices of the clique in reverse order. We consolidate this to:







$$\gamma^{\star}(C,\sigma)=\sum_{j=1}^{|C|}\Big| |\delta(C,\sigma)|+|C|+1-2j \Big|$$

Let us use $\kappa(C,\sigma)$ to denote $|C|+|\delta(C,\sigma)|$. If $\kappa(C,\sigma)$ is even:

$$\gamma^{\star}(C,\sigma)=\sum_{j=1}^{\min\left(|C|,\frac{\kappa(C,\sigma)}{2}\right)}\Big(\kappa(C,\sigma)+1-2j\Big) + \sum_{j=\frac{\kappa(C,\sigma)}{2}+1}^{|C|}\Big(2j-(\kappa(C,\sigma)+1)\Big)$$

We now claim that $\kappa(C,\sigma)/2 < |C|$ when $C$ is a large clique. Indeed, this follows from the fact that $\delta(C,\sigma) \leq |S| = k$, and further, $|C| > k$ since $C$ is a large clique. Therefore, the expression above simplifies to:

$$\gamma^{\star}(C,\sigma)=\sum_{j=1}^{\frac{\kappa(C,\sigma)}{2}}\Big(\kappa(C,\sigma)+1-2j\Big) + \sum_{j=\frac{\kappa(C,\sigma)}{2}+1}^{|C|}\Big(2j-(\kappa(C,\sigma)+1)\Big).$$

By a change of variable it's straightforward to check that the terms above evaluate to the following sums: 


$$\gamma^{\star}(C,\sigma)=\sum_{j=1}^{\frac{\kappa(C,\sigma)}{2}}\left(2j-1\right) +\sum_{j=1}^{\frac{|C|-|\delta(C,\sigma)|}{2}}\left(2j-1\right) = \frac{\delta(C,\sigma)^{2}+|C|^2}{2}.$$

Analogously\footnote{We skip directly to the last step since the argument is identical to the previous case.}, for the case when $\kappa(C,\sigma)$ is odd, we have:

$$\gamma^{\star}(C,\sigma)=\sum_{j=1}^{\frac{\kappa(C,\sigma)-1}{2}}\Big(\kappa(C,\sigma)+1-2j\Big) + \sum_{j=\frac{\kappa(C,\sigma)+1}{2}}^{|C|}\Big(2j-(\kappa(C,\sigma)+1)\Big)$$
$$=\sum_{j=1}^{\frac{\kappa(C,\sigma)-1}{2}}\left(2j\right) +\sum_{j=1}^{\frac{|C|-|\delta(C,\sigma)|-1}{2}}\left(2j\right) = \frac{\delta(C,\sigma)^{2}+|C|^{2}-1}{2}$$
Hence, 
$$\gamma^{+}(C,\sigma)=\left\lfloor\frac{\delta(C,\sigma)^{2}+|C|^2}{2}\right\rfloor - \left\lfloor\frac{|C|^2}{2}\right\rfloor = \begin{cases}
      \lfloor \frac{\delta(C,\sigma)^2}{2} \rfloor & \text{if } C \text{ is an even clique},\\
      \lceil  \frac{\delta(C,\sigma)^2}{2} \rceil & \text{if } C \text{ is an odd clique}.\\
    \end{cases}$$
This concludes the proof.
\end{proof}

\begin{corollary}
\label{cor:largecliques}
Let $G$ be a graph and let $S \subseteq V(G)$ be a twin cover of $G$ of size $k$. Further, let $C_i$ and $C_j$ be large cliques in $G \setminus S$ that have the same parity and type. If $\sigma$ is a layout that places $C_i$ and $C_j$ in the same location, then $\gamma^+(C_i) = \gamma^+(C_j)$.
\end{corollary}

\begin{proof}
The claim follows from the fact that $\delta(C_i,\sigma) = \delta(C_j,\sigma)$ when $C_i$ and $C_j$ are both large cliques with the same type and that share the same location in the layout $\sigma$. Further, $\gamma^+(C_i,\sigma) = \gamma^+(C_j,\sigma)$ since $C_i$ and $C_j$ are given to have the same parity. 
\end{proof}

Let us now formalize the notion of 'location' in an ordering. Let $\sigma$ be an arbitrary but fixed ordering. We say that a vertex $v$ in $G\setminus S$ is placed at the \emph{location} $|\sigma_{<v}\cap S| +1$, where $\sigma_{<v} := \{ w ~|~ w <_{\sigma} v\}$. If $\sigma$ is a clean ordering, we also say that a clique $C\in G\setminus S$ is placed at the \emph{location} $|\sigma_{<v}\cap S| +1$, where $v$ is an arbitrarily chosen vertex of $C$. Note that since $\sigma$ is a clean ordering, for any clique  $C\in G\setminus S$, all its vertices are placed at the same location. Therefore the notion of the location of a clique is well-defined. Intuitively, the location of a clique tells us where it lies in the layout relative to the twin cover vertices. In particular, cliques that are placed at location $1 < i \leq k$ lie between the $(i-1)^{th}$ and the $i^{th}$ twin cover vertex; with cliques at locations $1$ and $k+1$ being placed to the left of the first twin cover vertex and to the right of the last twin cover vertex, respectively. 

Let $\mathcal{C} :=\Big( 2^S \times \{0\} \times [k] \Big) \cup \Big( 2^S \times \{1\} \times \{e,o\} \Big)$. A \emph{class} is a triplet $(T,b,j) \in \mathcal{C}$. Recall that the type of a clique $C$ is given by $N(C) \cap S$. Let $T \subseteq S$ be arbitrary but fixed, and let $C$ be a clique of type $T$. Then, the class of the clique $C$ is given by:

\begin{itemize}
    \item $(T,1,o)$ if $C$ is a large odd clique.
    \item $(T,1,e)$ if $C$ is a large even clique.
    \item $(T,0,j)$ if $C$ is a small clique on $j$ vertices.
\end{itemize}

We will typically use $\nu$ to denote an element of $\mathcal{C}$. 

We now turn to the notion of specifications, which capture the ``demand'' that we may make for the number and the total sizes of the cliques of a particular class at a particular location. Formally, a specification is a map from $\mathcal{C} \times [k+1]$ to $[n] \cup \{0\}$. We relate specifications to layouts in the following definitions.

\begin{itemize}
\item Given a specification $\alpha$, we say that a layout $\sigma$ \emph{respects $\alpha$ in count} if, for each location $j \in [k+1]$ and for every class $\nu \in \mathcal{C}$, the number of cliques of class $\nu$ in location $j$ according to $\sigma$ is $\alpha(\nu,j)$. 
\item Given a specification $\beta$, we say that a layout $\sigma$ \emph{respects $\beta$ in size} if, for each location $j \in [k+1]$ and for every class $\nu \in \mathcal{C}$, the total size of cliques of class $\nu$ in location $j$ according to $\sigma$ is $\beta(\nu,j)$. 
\end{itemize}

We say that two layouts $\sigma$ and $\pi$ are \emph{similar with respect to $S$} if the layouts are identical when projected on the vertices of $S$. Our first observation is that the notion of specifications is sufficiently rich in the context of imbalance in the following sense: for an arbitrary but fixed pair of specifications $(\alpha,\beta)$, all similar layouts that respect $\alpha$ in size and $\beta$ in count have the same imbalance. We formalize this claim below.

\begin{lemma}
\label{lem:imbspecs}
Let $G$ be a graph and let $S$ be a twin cover of $G$ of size $k$. Also, let $\alpha$ and $\beta$ be two specification functions for $G$ and let $\sigma$ and $\pi$ be two clean layouts of $G$ that are similar with respect to $S$. If $\sigma$ and $\pi$ both respect $\alpha$ in count and $\beta$ in size, then $\mathcal{I}(\sigma) = \mathcal{I}(\pi)$. 
\end{lemma}

\begin{proof}
Label the vertices of $S$ in the order of their appearance in $\sigma$ as $s_1 < \ldots <s_k$. Note that since $\sigma$ and $\pi$ are similar, this is also the order of the vertices of $S$ in $\pi$. Let $S_0:=\phi$ and for every $i$ in $[k]$, let $S_i:=S_{i-1} \cup \{s_i\}$. Also, let:

$$\mathcal{L}:=\{C ~|~ C \text{ is a large clique in } G\setminus S\}.$$

For every $T\subseteq S$, $p\in [k]$, $q\in [k+1]$, we have the following. 

\begin{itemize}
    \item The imbalance of a clique $C$ of class $(T,0,p)$ (i.e, a small clique $C$ of type $T$ and size $p$) placed at location $q$ is given by:
    
    $$\gamma^{*}(C)=\sum_{j=1}^{p} \Big||T\cap (S\setminus S_{q-1})|-|T\cap S_{q-1}|+(p+1)-2j \Big| =: \theta_{(T,0,p)}(q)$$

    \item By Lemma~\ref{lem:largecliques}, the \emph{excess} imbalance of a clique $C$ of class $(T,1,e)$ (i.e, a large even clique $C$ of type $T$) placed at location $q$ is given by:
    
    $$\gamma^{+}(C)=\left\lfloor\frac{\Big(|T\cap(S\setminus S_{q-1})|-|T\cap S_{q-1}|\Big)^2}{2}\right\rfloor =: \theta_{(T,1,e)}(q)$$

    \item By Lemma~\ref{lem:largecliques}, the \emph{excess} imbalance of a clique $C$ of class $(T,1,o)$ (i.e, a large odd clique $C$ of type $T$) placed at location $q$ is given by:
    
    $$\gamma^{+}(C)=\left\lceil\frac{\Big(|T\cap(S\setminus S_{q-1})|-|T\cap S_{q-1}|\Big)^2}{2}\right\rceil =: \theta_{(T,1,o)}(q)$$
\end{itemize}
For every $i\in [k]$, imbalance of twin cover vertex $s_i$ is given by:
$$\mathcal{I}(s_i)=\left||N(s_i)\cap(S\setminus S_i)|-|N(s_i)\cap S_i|+\sum_{\substack{\nu = (T,\cdot,\cdot) \in \mathcal{C} \\ s_i \in T}} \Big(\sum_{q=i+1}^{k+1}\beta(\nu,q) - \sum_{q=1}^{i}\beta(\nu,q) \Big) \right|.$$

Hence, the imbalance of \emph{any} layout that respects $\alpha$ in count and $\beta$ in size, and orders the vertices of the twin cover as stated in the beginning is given by:
$$\sum_{i=1}^{k}\mathcal{I}(s_i) + \sum_{T\in 2^S}\sum_{q=1}^{k+1}\sum_{p=1}^{k}\Big(\alpha((T,0,p),q)\cdot\theta_{(T,0,p)}(q)\Big) + $$
$$\sum_{C\in \mathcal{L}}\gamma(|C|) +\sum_{T\in 2^S}\sum_{q=1}^{k+1}\Big(\alpha((T,1,e),q)\cdot\theta_{(T,1,e)}(q) + \alpha((T,1,o),q)\cdot\theta_{(T,1,o)}(q)\Big).$$

Since these criteria apply to both $\sigma$ and $\pi$, it follows that they have the same imbalance and their imbalance is given by the expression above. 
\end{proof}

Based on the proof of Lemma~\ref{lem:imbspecs}, we have the following lemma.

\begin{lemma}
\label{lem:compversion}
Let $G$ be a graph and let $S$ be a twin cover of $G$ of size $k$. Also, let $\alpha$ and $\beta$ be two specification functions for $G$. Let $\pi$ be an ordering of the vertices of $S$. For any layout $\sigma$ of the vertices of $G$ that respects $\alpha$ in count and $\beta$ in size, and which is consistent with $\pi$ when restricted to $S$, its imbalance can be computed in time $O(g(k)\cdot k \cdot n^{O(1)})$.
\end{lemma}

\begin{proof}
Given $\alpha$, $\beta$ and $\pi$, the lemma follows from the fact that the imbalance of any layout $\sigma$ of the vertices of $G$ that respects $\alpha$ in count and $\beta$ in size, and which is consistent with $\pi$ when restricted to $S$ is given by:
$$\sum_{i=1}^{k}h(i) + \sum_{T\in 2^S}\sum_{q=1}^{k+1}\sum_{p=1}^{k}\Big(\alpha((T,0,p),q)\cdot\theta_{(T,0,p)}(q)\Big) + $$
$$\sum_{C\in \mathcal{L}}\gamma(|C|) +\sum_{T\in 2^S}\sum_{q=1}^{k+1}\Big(\alpha((T,1,e),q)\cdot\theta_{(T,1,e)}(q) + \alpha((T,1,o),q)\cdot\theta_{(T,1,o)}(q)\Big),$$
where:
\begin{itemize}
\item $S_0:=\phi$, and for each $i\in[k], S_i=S_{i-1}\cup \{\pi(i)\}$
\item for each $T\in 2^{S}, p\in[k],q\in[k+1]$,
$$\theta_{(T,0,p)}(q) := \sum_{j=1}^{p} \Big||T\cap (S\setminus S_{q-1})|-|T\cap S_{q-1}|+(p+1)-2j \Big|$$
\item for each $T\in 2^{S}, q\in[k+1]$,
$$\theta_{(T,1,e)}(q) := \left\lfloor\frac{\Big(|T\cap(S\setminus S_{q-1})|-|T\cap S_{q-1}|\Big)^2}{2}\right\rfloor$$
$$\theta_{(T,1,o)}(q) := \left\lceil\frac{\Big(|T\cap(S\setminus S_{q-1})|-|T\cap S_{q-1}|\Big)^2}{2}\right\rceil$$
\item for each $i\in[k]$,
$$h(i) := \left||N(\pi(i))\cap(S\setminus S_i)|-|N(\pi(i))\cap S_i|+\sum_{\substack{\nu = (T,\cdot,\cdot) \in \mathcal{C} \\ \pi(i) \in T}} \Big(\sum_{q=i+1}^{k+1}\beta(\nu,q) - \sum_{q=1}^{i}\beta(\nu,q) \Big) \right|.$$
\end{itemize}
It is easy to check that the expression above can be computed in time $O(g(k)\cdot k \cdot n^{O(1)})$. 
\end{proof}

Our next claim is that the number of specification functions is bounded as a function that is XP in $k$. More specifically, we have the following. 

\begin{proposition}
Let $G$ be a graph and let $S$ be a twin cover of $G$ of size $k$. Then, the number of specification functions is bounded by $(n+1)^{g(k)}$, where $g(k) = (2^k \cdot (k+2)) \times (k+1)$.
\end{proposition}

\begin{proof}
This follows from the fact that the number of possible classes is at most $(2^k \cdot (k+2))$, and that the number of functions from a domain with $a$ elements to a range with $b$ elements is $b^a$.
\end{proof}

We are now finally ready to present our dynamic programming algorithm. For any pair of specifications $(\alpha,\beta)$, we say that a layout respects $(\alpha,\beta)$ if it respects $\alpha$ in count and $\beta$ in size. Recall that the cliques of $G \setminus S$ were denoted by $C_1, \ldots, C_r$. For $j \in [r]$, let $H_j$ denote the graph $G[S \cup C_1\cup \ldots\cup C_j]$. Now consider the following DP table, where $\alpha$ and $\beta$ are specifications and $q \in [r]$:

 \begin{equation*}
    \mathbb{T}(\alpha,\beta,q)  =
    \begin{cases}
      1 & \text{if there exists a layout } \sigma \text{ of } H_q \text{ that respects } (\alpha,\beta),\\
      0 & \text{otherwise}.
    \end{cases}
\end{equation*}

Before describing the recurrence for $\mathbb{T}_\pi(\alpha,\beta,q)$, we informally allude to why this is useful to compute. To check if $G$ admits a layout of imbalance at most $t$, our algorithm proceeds as follows. For all specification pairs $(\alpha,\beta)$, we check if  $\mathbb{T}(\alpha,\beta,r) = 1$. For all the instances where the entries are one, we compute the imbalance of any layout that respects $(\alpha,\beta)$ based on Lemma~\ref{lem:imbspecs}, by trying all possible choices for the ordering of twin cover vertices. If we ever encounter an imbalance value that is at most $t$ then we abort and return \textsc{Yes}, otherwise we return \textsc{No} after all choices of $\pi$ and the corresponding specification pairs have been exhaustively examined. 

We now turn to the computation of the DP table $\mathbb{T}$. For the base case, we have $q = 1$, and it is easy to see that there are exactly $(k+1)$ choices --- one for each possible location --- of pairs of specifications $(\alpha,\beta)$ for which $\mathbb{T}_\pi(\alpha,\beta,q) = 1$. For the sake of exposition, we explicitly describe these pairs. Recall that the size of $C_1$ is given by $\ell_1$, suppose the class of $C_1$ is $\mathfrak{C}^\star$. Then consider the specification function pairs $(\alpha_i,\beta_i)_{i \in [k+1]}$ defined as follows:

 \begin{equation*}
    \alpha_i(\mathfrak{C},j) =
    \begin{cases}
      1 & \text{if } j = i \text{ and } \mathfrak{C} = \mathfrak{C}^\star,\\
      0 & \text{otherwise},
    \end{cases}
\end{equation*}

and

 \begin{equation*}
    \beta_i(\mathfrak{C},j) =
    \begin{cases}
      \ell_1 & \text{if } j = i \text{ and } \mathfrak{C} = \mathfrak{C}^\star,\\
      0 & \text{otherwise}.
    \end{cases}
\end{equation*}

This motivates the definition of $\mathbb{T}$ for the base case:

\begin{equation*}
    \mathbb{T}(\alpha,\beta,1)  =
    \begin{cases}
      1 & \text{if there exists } i \in [k+1] \text{ such that } \\
      ~ & \alpha = \alpha_i \text{ and } \beta = \beta_i \\
      0 & \text{otherwise}.
    \end{cases}
\end{equation*}

Before proceeding to the recurrence, let us introduce a definition that will make the recurrence simpler to describe. We say that a clique $C$ of size $\ell$ whose class is $\mathfrak{C}$ is an \emph{overfit} for a location $j \in [k+1]$ with respect to the pair of specifications $(\alpha,\beta)$ if:

\begin{itemize}
    \item $\alpha(\mathfrak{C},j) = 0$, i.e, there is no ``demand'' for a clique of class $\mathfrak{C}$ at location $j$; \emph{or}
    \item $\beta(\mathfrak{C},j) < \ell$, i.e, the total sizes of the cliques of class $\mathfrak{C}$ that are expected at location $j$ is smaller than the size of $C$.
\end{itemize}

We now turn to the recurrence for $\mathbb{T}(\alpha,\beta,q)$ for some $q \in [r]$. Let the class of the clique $C_q$ be $\mathfrak{C}^\star$. Define the following auxiliary specifications for $i \in [k+1]$, which, intuitively speaking, capture the subproblems of interest if the clique $C_q$ were to be placed at location $i$:

 \begin{equation*}
    \alpha_i(\mathfrak{C},j) =
    \begin{cases}
      \alpha(\mathfrak{C},j)-1 & \text{if } j = i \text{ and } \mathfrak{C} = \mathfrak{C}^\star,\\
      \alpha(\mathfrak{C},j) & \text{otherwise},
    \end{cases}
\end{equation*}

and

 \begin{equation*}
    \beta_i(\mathfrak{C},j) =
    \begin{cases}
      \beta_i(\mathfrak{C},j) - \ell_q & \text{if } j = i \text{ and } \mathfrak{C} = \mathfrak{C}^\star,\\
      \beta_i(\mathfrak{C},j) & \text{otherwise}.
    \end{cases}
\end{equation*}

Let $B \subseteq [k+1]$ be the set of all locations for which $C_q$ is \emph{not} an overfit with respect to $(\alpha,\beta)$. Then, we have:

$$ \mathbb{T}(\alpha,\beta,q) = \lor_{i \in B} \mathbb{T}(\alpha_i,\beta_i,q-1).$$

The discussions above lead us to the main result of this section.

\begin{theorem}
\label{thm:xp}
\textsc{Imbalance} is in XP when parameterized by twin cover. 
\end{theorem}

\begin{proof}
Our algorithm begins by computing $\mathbb{T}$ as described above. For every pair of specification functions $(\alpha,\beta)$ such that $\mathbb{T}(\alpha,\beta,r) = 1$, we guess a layout $\pi$ of the vertices of the twin cover $S$ and use Lemma~\ref{lem:imbspecs} to compute the imbalance of any layout which respects $(\alpha,\beta)$ and is consistent with $\pi$ when restricted to $S$. Observe that the DP table has $((n+1)^{2g(k)} \cdot r)$ entries, each of which require $O(k)$ table lookups to be computed. Therefore, the total running time of this approach, ignoring factors polynomial in $n$, is given by $O^\star(k! \cdot (n+1)^{2g(k)} \cdot r \cdot k \cdot f(k))$, where $f(k)$ is the time required to compute the expression given by Lemma~\ref{lem:compversion}. 

It only remains to establish the correctness of this approach. The correctness of the recurrence defining $\mathbb{T}$ is an immediate consequence of the definitions. We now turn to the correctness of the overall algorithm. 

On the one hand, if the input is a \textsc{Yes} instance, then there exists a clean ordering of $V(G)$, say $\sigma$, whose imbalance is at most $t$. Let $\pi$ be the ordering obtained from $\sigma$ restricted to the twin cover vertices, and let $(\alpha,\beta)$ be the unique pair of specification functions that $\sigma$ respects. Note that $\mathbb{T}(\alpha,\beta,r) = 1$, and the imbalance given by Lemma~\ref{lem:compversion} when applied with $(\alpha,\beta)$ as the choice of specification functions and $\pi$ as the relative order on $S$ is exactly the imbalance of $\sigma$. Therefore, our algorithm outputs \textsc{Yes}. 

Conversely, if the output of our algorithm is positive, then there is a choice of $\pi$, an ordering of the twin cover vertices, and pair of specification functions $(\alpha,\beta)$, for which:

\begin{itemize}
    \item The entry $\mathbb{T}(\alpha,\beta,r) = 1$, and,
    \item the imbalance as given by Lemma~\ref{lem:imbspecs} when applied with $(\alpha,\beta)$ as the choice of specification functions and $\pi$ as the relative order on $S$ is at most $t$.
\end{itemize}

By the semantics of $\mathbb{T}$ and the correctness of the recurrence, we know that $\mathbb{T}(\alpha,\beta,r) = 1$ implies the existence of a layout $\sigma$ of $V(G)$ that respects $(\alpha,\beta)$. We rearrange the vertices of $S$ so that the final layout is consistent with $\pi$. Observe that Lemma~\ref{lem:compversion} guarantees that the imbalance of the modified layout is at most $t$, and this concludes our argument. 
\end{proof}

To conclude this section, we also remark that \textsc{Imbalance} is in NP for succinct inputs.

\begin{proposition}
\label{prop:np}
\textsc{Imbalance} is in NP for succinct inputs. 
\end{proposition}

\begin{proof}
We argued the correctness of above claim for twin cover of size one in Section $3$.Now consider the case of a twin cover $S \subseteq V(G)$ of any size $k$, not necessarily one. If the input $(G,S,t)$ is a YES-instance,we know that there exists a clean ordering,say $\sigma$, of imbalance at most $t$. Since the imbalance of any clean ordering is completely determined by the relative order of twin cover vertices and the location of every clique of $G\setminus S$ in the ordering, it suffices to provide $\sigma$ when restricted to $S$ (say $s_1<\ldots<s_k$) and the location of every clique $C_1,....,C_r$ of $G\setminus S$ in $\sigma$ as a certificate.Let $S_0:=\phi$ and for every $i$ in $[k]$, let $S_i:=S_{i-1}\cup\{s_i\}$. For every $i$ in $[r]$, let us denote the location of $C_i$ as $d[i]$. For every clique $C$ of $G\setminus S$, let $\tau(C)$ denote the type of $C$. Also, for every $q$ in $[k+1]$, let $P_q := \{C_j ~|~ j\in[r], d[j]=q\}$, i.e., $P_q$ denotes the set of cliques of $G\setminus S$ that are placed at location $q$ in $\sigma$.
Given the certificate,one can compute the imbalance of $\sigma$ and compare it against the target, i.e., $t$. This computation can be done in strongly polynomial time as follows:

\begin{itemize}
    \item For every $i$ in $[k]$, the imbalance of $i^{\text{th}}$ twin cover vertex, $s_i$ can be computed as follows:
$$\mathcal{I}(s_i)= \left| \sum_{q=1}^{i} \sum_{\substack{C \in P_q \\ s_i \in \tau(C)}} |C|- \sum_{q=i+1}^{k+1} \sum_{\substack{C \in P_q \\ s_i \in \tau(C)}} |C| \right| $$
    \item For every $i$ in $[r]$, imbalance of the vertices of clique $C_i$ can be computed as follows:
    \begin{itemize}
        \item if $C_i$ is a small clique,
$$\gamma^{\star}(C_i) = \sum_{j=1}^{|C_i|} \Big| |\tau(C_i)\cap(S\setminus S_{d[i]-1}) |-|\tau(C_i)\cap S_{d[i]-1}| + (|C_i|-j)-(j-1) \Big|$$
        \item if $C_i$ is a large even clique,
$$\gamma^{\star}(C_i)=\gamma(|C_i|)+ \left\lfloor\frac{(|\tau(C_i)\cap(S\setminus S_{d[i]-1}) |-|\tau(C_i)\cap S_{d[i]-1}|)^2}{2} \right\rfloor$$
        \item if $C_i$ is a large odd clique,
$$\gamma^{\star}(C_i)=\gamma(|C_i|)+ \left\lceil\frac{(|\tau(C_i)\cap(S\setminus S_{d[i]-1}) |-|\tau(C_i)\cap S_{d[i]-1}|)^2}{2} \right\rceil$$
    \end{itemize}
\end{itemize}
Now, the imbalance of the layout can be computed as:
$$\sum_{i=1}^{k} \mathcal{I}(s_i) + \sum_{i=1}^{r} \gamma^{\star}(C_i).$$
This concludes the proof.
\end{proof}

\nmtodo{Maybe introduce a computational version of Lemma 4 to make this discussion a little more fluid.}

\section{Parameterizing by bounded Twin Cover}
\label{sec:ilp}
In this section, we describe an ILP approach for the imbalance problem parameterized by twin cover when the entire input is given explicitly in the standard form. As in the previous section, we use $G$ to denote the graph given as input, $S \subseteq V(G)$ denotes a twin cover of size $k$. Recall that the question we are addressing is if $G$ admits a layout whose imbalance is at most $t$. We denote the cliques of $G \setminus S$ by $C_1, \ldots, C_r$, and we let $\ell_i$ denote the number of vertices in $C_i$. We also use $\ell$ to denote $\max\{\ell_1, \ldots, \ell_r\}$. Further, for each $\tau\in 2^{S}, p\in[\ell]$, we use $W(\tau,p)$ to denote the number of $p$-sized cliques of type $\tau$. The number of variables that we introduce here will be a function of $k$ and $\ell$.

To begin with, let $\pi$ be an arbitrary but fixed order on the twin cover vertices. Label the vertices of $S$ in the order of their appearance in $\pi$ as $s_1 < \ldots <s_k$. Let $S_0:=\phi$ and for every $i$ in $[k]$, let $S_i:=S_{i-1} \cup \{s_i\}$. We also guess a sign signature $t_j \in \{+1,-1\}$ for all $j \in [k]$ which we will need for technical reasons that will be apparent in a moment. Now, for each $\tau \in 2^S, q \in [k+1], p \in [\ell]$, we introduce a variable $x_{\tau,p}^q$ for denoting the number of cliques of size $p$ and type $\tau$ placed at location $q$ in a layout. 

In particular, let $X := \{x_{\tau,p}^q ~|~ \tau \in 2^S, q \in [k+1], p \in [\ell]\}.$ We say that an assignment $f: X \rightarrow [n] \cup \{0\}$ is \emph{valid} if:

$$ \text{for each } \tau \in 2^S, p \in [\ell]: \sum_{q = 1}^{k+1} x_{\tau,p}^q = W(\tau,p). $$

For a valid assignment $f$, we let $\sigma_f$ be the following layout of $V(G)$: for each $q \in [k+1]$, $\tau \in 2^S$ and $p \in [\ell]$, we place $x_{\tau,p}^q$ $p$-sized cliques of type $\tau$ at location $q$. Further, we arrange the twin cover vertices according to $\pi$. Now, for $i \in [k]$, we define:

$$R(i) :=  \big| N(s_i) \cap (S \setminus S_i) \big|  + \left(\sum_{q=i+1}^{k+1} \sum_{\substack{\tau \in 2^S \\s.t.\\ s_i \in \tau}} \sum_{p = 1}^{\ell} (p \cdot x^q_{\tau,p})\right),$$

and:

$$L(i) :=  \big| N(s_i) \cap S_{i-1} \big|  + \left(\sum_{q=1}^{i} \sum_{\substack{\tau \in 2^S \\s.t.\\ s_i \in \tau}} \sum_{p = 1}^{\ell} (p \cdot x^q_{\tau,p})\right).$$

Given a valid assignment $f$, note that $L(i)$ and $R(i)$ give the number of left and right neighbors of $s_i$ if each $x^q_{\tau,p}$ is substituted with $f(x^q_{\tau,p})$.

Finally, for each $\tau \in 2^S$ and $q \in [k+1]$, we define $\Delta_\tau^q := \big| \tau \cap (S \setminus S_{q-1}) \big| - \big| \tau \cap S_{q-1} \big|$ and let:

$$c_{\tau,p}^q := \left( \sum_{i = 1}^{p}  \Big| \Delta^q_\tau + (p - i) - (i - 1) \Big|  \right).$$

We are now ready to describe the ILP formulation, which we propose as follows. 

Minimize:

$$\sum_{q=1}^{k+1} \sum_{\tau \in 2^S} \sum_{p = 1}^{\ell} c_{\tau,p}^q \cdot x_{\tau,p}^q  + \sum_{j=1}^{k} t_j (R(j) - L(j)) $$


subject to:

$$\text{for each } \tau \in 2^S, p \in [\ell]: \sum_{q = 1}^{k+1} x_{\tau,p}^q = W(\tau,p),$$

$$\text{for all } j \in [k]: t_j (R(j) - L(j)) \geq 0, $$

and:

$$x_{\tau,p}^q \geq 0 \text{ for all } \tau \in 2^S, p \in [\ell], \mbox{ and } q \in [k+1].$$

Note that the number of variables in the ILP is given by $2^k \cdot (k+1) \cdot \ell$, which implies that the time we need to solve the ILP using Lenstra's algorithm is FPT in $k$ and $\ell$. The correctness of the formulation is largely self-evident, and we sketch a brief argument for the sake of completeness.

Any clean ordering $\sigma$ of $V(G)$ can be realized by an assignment to the variables of the ILP as follows: for each $q \in [k+1], \tau \in 2^S, p \in [\ell]$, we set $x_{\tau,p}^q$ to the number of $p$-sized cliques of type $\tau$ that are placed in location $q$ of $\sigma$. Further, let $\pi$ be the ordering $\sigma$ restricted to $S$. For each $j\in[k]$, we set $t_j = 1$ if $N_R(s_j,\sigma) \geq N_L(s_j,\sigma)$, and $t_j = -1$ otherwise. It is easy to check that the objective function of the ILP evaluates to $\mathcal{I}(\sigma)$ under this assignment of the variables. Since there exists an optimal clean ordering, we have that the ILP is minimized at a value that is at most the optimal imbalance of $G$ among all layouts that order the vertices of the twin cover according to $\pi$.

In the other direction, consider an assignment $g$ to the variables $X \cup \{t_j ~|~ j \in [k]\}$ that optimizes the ILP above. Let $f$ be the restriction of $g$ to $X$. By the first constraint, we know that $f$ is a valid assignment as defined before. Note that the setting of the values of $t_j$'s mimic the behavior of the absolute value function due to the second constraint. Based on this, it is straightforward to verify that the value of the objective function with respect to the assignment $f$ corresponds to the imbalance of the layout $\sigma_f$. This establishes that the optimal imbalance of $G$ among all layouts that order the vertices of the twin cover according to $\pi$ is also a lower bound for the optimal value of the ILP.

Therefore, by examining all $k!$ orderings of the twin cover vertices $S$ and solving the corresponding ILPs, it is clear from the discussion above that we can determine the optimal imbalance of $G$ in FPT time when parameterized by the size of the twin cover $S$ and size of the largest clique in $G \setminus S$. This concludes the proof of the following theorem.

\begin{theorem}
\label{thm:fpt}
\textsc{Imbalance} is FPT when parameterized by $(k+\ell)$, where $k$ is the size of a $\ell$-bounded twin cover. 
\end{theorem}

\section{Concluding Remarks}
\label{sec:concl}
We investigated the complexity of \textsc{Imbalance} parameterized by twin cover. We demonstrated that that the problem is XP by a dynamic programming approach, and that it is FPT when parameterized by twin cover as well as the size of the largest clique in the graph when the twin cover is removed. It is also easy to see that the problem is FPT when parameterized by the twin cover and the \emph{number} of cliques outside the twin cover, by simply restricting our attention to clean orderings and trying all possible permutations of the cliques and guessing where the twin cover vertices `fit' among them. This leads us to conclude that the tractable cases are, roughly speaking, when there are a small number of large cliques or a large number of small cliques, and the interesting cases lie in the middle of this spectrum. The most evident open question that emerges from our discussions is the issue of whether \textsc{Imbalance} is FPT when parameterized by twin cover.

We also introduced a notion of succinct representations of graphs in terms of their twin cover. It would be interesting to revisit problems which are FPT in twin cover but with a ``pseudo-polynomial'' running time in the setting of succinct input as described here. In particular, from the work of Ganian~\cite{Ganian11}, there are already some problems whose stated algorithms are not efficient in the succinct setting. For example, the algorithm for~\textsc{Boxicity} relies on the fact that the edges within the clique are irrelevant, based on which we may obtain an equivalent instance for which the twin cover becomes a vertex cover. It would be an interesting direction of future work to examine which of these problems remain FPT with a strongly polynomial running time when we work with succinct representations. Our work here demonstrates that \textsc{Imbalance} is already one problem for which the representation has a non-trivial influence on the complexity.

%
%

\bibliography{mybibliography}

\end{document}